\documentclass{gITR2e}

 \usepackage{epstopdf}
\usepackage{subfigure}

\usepackage{amssymb, color}
\usepackage{mathrsfs}
\usepackage{amsmath}
\usepackage{amsthm}
\usepackage{url}
\usepackage{stackrel} 
\usepackage{courier}
\usepackage[all,cmtip]{xy} 
\usepackage{multicol}
\usepackage{mathtools}
\usepackage{verbatim} 

\DeclarePairedDelimiter\floor{\lfloor}{\rfloor}

\catcode`~=11 
\newcommand{\urltilde}{\kern -.15em\lower .7ex\hbox{~}\kern .04em}  
\catcode`~=13 

\newcommand{\Der}{\mathop{\mathrm{Der}}\nolimits}

\newcommand{\Aut}{\mathop{\mathrm{Aut}}\nolimits}

\theoremstyle{plain}
\newtheorem{theorem}{Theorem}[section] 
\newtheorem{proposition}[theorem]{Proposition}
\newtheorem{lemma}[theorem]{Lemma}

\newtheorem{definition}[theorem]{Definition}
\newtheorem{remark}[theorem]{Remark}

\newtheorem{corollary}[theorem]{Corollary}

\begin{document}

\title{{\itshape Families of Orthogonal Laurent Polynomials, Hyperelliptic Lie algebras and Elliptic Integrals}} 
\author{\name{Ben Cox$^{\rm a}$$^{\ast}$\thanks{$^\ast$Corresponding author. Email: coxbl@cofc.edu}  and Mee Seong Im$^{b}$$^{\ast\ast}$\thanks{$^{\ast\ast}$ Email: meeseong.im@usma.edu}}
\affil{$^{a}$Department of Mathematics, College of Charleston, Charleston, SC 29424, USA; \\ $^{b}$Department of Mathematical Sciences, United States Military Academy, \\ West Point, NY 10996, USA}}

\maketitle 

\begin{abstract}
We describe a family of polynomials discovered via a particular recursion relation, which have connections to Chebyshev polynomials of the first and the second kind, and the polynomial version of Pell's equation. Many of their properties are listed in Section~\ref{section:properties-results}.  
We show that these families of polynomials in the variable $t$ satisfy certain second order linear differential equations that may be of interest to mathematicians in conformal field theory and number theory. 
We also prove that these families of polynomials in the setting of Date-Jimbo-Kashiwara-Miwa algebras when multiplied by a suitable power of $t$ are orthogonal with respect to explicitly-described kernels. Particular cases lead to new identities of elliptic integrals (see Section~\ref{section:Date-Jimbo-Kashiwara-Miwa}). 
\end{abstract}  

\begin{keywords} Hyperelliptic Lie algebras;
Krichever-Novikov algebras; universal central extensions; Date-Jimbo-Kashiwara-Miwa algebras; elliptic integrals; Pell's equation; Chebyshev polynomials
\end{keywords}

\begin{classcode}17B65, 33C47; 14H55; 33E05; 17B05 \end{classcode}

\setcounter{tocdepth}{0} 


\section{Introduction}\label{section:intro}

Let $R$ be the ring of meromorphic functions on a Riemann surface and with a fixed finite number of poles.
Krichever-Novikov algebras arise as central extensions of Lie algebra of derivations $\Der(R)$,  of loop algebras $\mathfrak{g}\otimes R$, and Lax algebras (see \cite{MR925072}, \cite{MR902293}, \cite{MR998426},  \cite{MR1666274}, \cite{MR1989644}, \cite{MR2451204}, \cite{MR2985911}).  Let $p(t)\in \mathbb C[t]$ be a polynomial with distinct complex roots. 
In this paper, we concentrate on the ring $R_2(p)$, 
where 
$R_m(p)=\mathbb{C}[t^{\pm 1}, u]/\langle  u^m - p(t) \rangle$ is the coordinate ring of a Riemann surface (when $m=2$) with a finite number of punctures.  
The universal central extension of the Lie algebra $\mathcal{R}_m(p)=\Der(R_m(p))$ is called the $m$-th superelliptic Lie algebra associated to $p$ (cf. \cite{cox2014simple}), and is a particular type of Krichever-Novikov algebra.   When $m=2$ and $p$ is separable of degree greater than $4$, the Lie algebras of derivations $\Der(R_2(p))$ and of loop algebras $\mathfrak{g}\otimes R_2(p)$ are hyperelliptic.  

In previous work of the first author with X. Guo, R. Lu and K. Zhao, interesting automorphism groups of $\text{Der}(R)$ appear when $R=\mathbb C[t,(t-a_1)^{-1},\dots, (t-a_n)^{-1}]$.  In particular, the five families of groups studied by Klein ($C_n$, $D_n$, $A_4$, $S_4$ and $A_5$) are precisely the only groups that appear as automorphism groups of derivations of $R$ (cf. \cite{MR3211093}). Automorphism groups play an important role in the study of conformal field theory.  For example, the monster simple group is known to appear as the automorphism group of a particular vertex operator algebra (cf. \cite{MR996026}, \cite{MR849120}, \cite{MR781381}), and other finite sporadic simple groups make their appearance as conjectured automorphism groups in Umbral Moonshine. 

In this paper, we study families of polynomials that arise in the description of the groups of units of the automorphism group $\Aut(\mathcal{R}_2(p))$ of the Lie algebra of derivations of $R_2(p)$. The original motivation for describing these units is to help one determine the automorphism ring and consequently, the classification problem of the rings $R_2(p)$, and thus Lie algebras $\mathcal{R}_2(p)$. Such invertible elements have deep and important connections to the solutions of the polynomial Pell equation 
$$
f^2-g^2p=1, \quad f,g\in \mathbb C[t]
$$ that appear in number theory (cf. \cite{cox2014simple} Lemma 11, \cite{MR2183270}). 
 In the case when $p=t^2-1$, the pairs of Chebyshev polynomials $T_n(t)$ and $U_n(t)$ of the first kind and of the second kind, respectively, provide us with the solutions to the particular Pell equation 
 \begin{equation}
 T_n(t)^2-U_{n-1}(t)^2(t^2-1)=1
 \end{equation}
(cf. \cite{MR1937591}),  and thus give a complete description of the group of units.  
  
We would like to give another motivation for the study of such rings and their associated Lie algebras $\text{Der}(\mathcal R_{2}(p))$ and $\mathfrak g\otimes R_2(p)$.  Date-Jimbo-Kashiwara-Miwa studied integrable systems arising from Landau-Lifshitz differential equation in \cite{MR701334}. 
This differential equation describes time evolution of magnetism in solids: 
\begin{equation}\label{equation:Landau-Lifshitz-differential-equation}
\mathbf{S}_t = \mathbf{S} \times \mathbf{S}_{xx} + \mathbf{S}\times \mathbf{JS}, 
\end{equation}
where 
\[ 
\mathbf{S} = (S_{1}, S_{2}, S_{3}),  
	\hspace{6mm} 
S_{1}^2 + S_{2}^2 + S_{3}^2 = 1,  
	\hspace{6mm} 
\mathbf{J} = 
	\begin{pmatrix}
		J_1 & 0 & 0 \\
		0 & J_2 & 0 \\ 
 		0 & 0 & J_3 \\ 
	\end{pmatrix}, 
	\hspace{6mm}
	J_i \in \mathbb{C}. 
\] 
 The Landau-Lifshitz (LL) differential equation has been integrated in \cite{Borovik}. 
Note that a Lax pair for the Landau-Lifshitz differential equation was found, and its elliptic automorphic Lie algebras were introduced in \cite{sklyanin-complete-integrability}.  Here 
\begin{equation}
\frac{\partial W}{\partial x_1}=LW,\quad \frac{\partial W}{\partial x_2}=MW
\end{equation}
where $x_1=x,x_2=-\imath t$ and 
\begin{align*}
L&:=\sum_j=1^3z_jS_j\sigma_j, \\
M&:=\imath \sum_{j,k,l=1}^3z_j\sigma_jS_kS_l\epsilon^{jkl}+2z_1z_2z_3\sum_{j=1}^3z_j^{-1}S_j\sigma_j.
\end{align*}
and the $\sigma_j$ are the usual Pauli spin matrices.  The coordinates $(z_1,z_2,z_3)$ are the coordinates of an elliptic curve $z_i^2-z_j^2=\frac{1}{4}(J_i-J_j)$ where $1\leq i,j\leq 3$.  
Besides this the corresponding Riemann-Hilbert problem for the LL equation was developed in
\cite{Mikhailov-Landau-Lifshitz-eqn}
and in \cite{Rodin-Riemann-boundary-problem}.
The resulting LL hierarchy and its appearance in the study of the asymmetric chiral field was also treated in
\cite{Holod-hidden-symmetry}.

From an unramified covering sending $(z_1,z_2,z_3)$ to $(u=4z_1z_2,t=2z_3)$, Date-Jimbo-Kashiwara-Miwa introduced \color{black}  in \cite{MR823315} the infinite-dimensional Lie algebra $\widehat{\mathfrak{sl}(R_2(p))}$ which is a one-dimensional central extension of 
\begin{equation} 
\mathfrak{g}\otimes \mathbb{C}[t^{\pm1}, u: u^2 = (t^2-b^2)(t^2-c^2)],  \label{DJKMalgebra}
\end{equation}
where $b,c\in \mathbb{C}$, $b\not= \pm c$, and $\mathfrak{g}$ is a simple finite-dimensional Lie algebra. This central extension acts on the solutions of~\eqref{equation:Landau-Lifshitz-differential-equation} 
as infinitesimal B\"acklund transformations, which is the motivation for calling these algebras $DJKM$-algebras. One should note that they are particular examples of Krichever-Novikov algebras. 
After a suitable change of variables, we reduce the study of this algebra to $R_2(p)$, where 
\begin{equation}\label{eq:polynomial-DJKM}
p(t)= p_{\beta}(t) = \dfrac{t^4-2\beta t^2+1}{\beta^2-1}, \hspace{4mm} \beta \not=\pm 1. 
\end{equation} 
As the main example in this paper focuses on a description of units in ring $R_2(p)$ for the polynomial $p$ above (and other separable polynomials $p$), we will not completely review how certain nonclassical orthogonal polynomials appear in the study of the universal central extension of the Lie algebra $\text{Der}(R_2(p))$ and the loop algebra $\mathfrak{sl}_2\otimes R_2(p)$ for the polynomial in  \eqref{eq:polynomial-DJKM}. Thus, we will be content with noting that interesting families such as associated Legendre, associated Jacobi, ultraspherical, and Chebyshev polynomials arise in the description of the universal central extension of $\text{Der}(R_2(p))$ (cf. \cite{coxzhao-2015}) and the loop algebra $\mathfrak{sl}_2\otimes R_2(p)$ (cf. \cite{MR3090080}).
 Also it should be noted that particular examples of the associated Jacobi polynomials of Ismail-Wimp make their appearance in satisfying certain fourth order linear differential equations (\cite{MR1669843}). Lastly, we point out that the fourth order differential equation in \cite{MR3090080} also seems to be related to Kaneko-Zagier's work on supersingular $j$-invariants and Atkins polynomials (see \cite{MR1486833}). 


One should also consider in this context the Krichever-Novikov equation as developed in the two papers \cite{Krichever-Novikov-algebraic-curves} and \cite{Krichever-Novikov-holomorphic-bundles-nonlinear-equations}.  
A Lax pair \'a la Sklyanin, in terms of the elliptic automorphic Lie algebras, was studied in \cite{Guil-Manas}.
A symmetry analysis of the Krichever-Novikov equation was developed in \cite{Svinolupov-Sokolov-Yamilov}. \color{black}

A natural generalization of DJKM algebra \eqref{DJKMalgebra} and $\text{Der}(R)$ is where the coordinate ring of the elliptic curve $R= \mathbb{C}[t^{\pm1}, u: u^2 = (t^2-b^2)(t^2-c^2)]$ is replaced by a hyperelliptic curve $R_2(p)= \mathbb{C}[t^{\pm1}, u: u^2 = p(t)]$ where $p(t)$ is a separable polynomial of degree greater than $4$.  
In \cite{cox2014simple}, it was given necessary and sufficient conditions for the Lie algebra $\text{Der}(R_m(p))$ to be simple, which their universal central extensions and their derivation algebras have been explicitly described. 
The authors also studied the isomorphism and automorphism problem for these Lie algebras by describing the group of units of $R_2(p)$ for particular $p$.   In the process, it was realized that the group of units consists of sums of the form $f_n+g_n\sqrt{p}$,  where $f_n$ and $g_n$ are polynomials in $t$, $\sqrt{p}$ is $\pm u$, and $f_n$ and $g_n$ satisfy the polynomial version of Pell's equation $f_n^2-g_n^2p=ct^k$ for some $k\in\mathbb Z$ and $c\in \mathbb C^\times$ (cf. Lemma \ref{lemma:multiplicative-group}).   The hyperelliptic Lie algebras are particular types of Krichever-Novikov (KN) algebras and some of the representation theory of such KN algebras have been reviewed and described in the monographs \cite{MR2985911} and \cite{MR2451204}.\color{black}

In this paper, we use the generating series for pairs of polynomials $a_n$ and $b_n$ satisfying the polynomial version of Pell's equation to derive a recurrence relation for $a_n$ and $b_n$ where $a_n+b_n\sqrt{p}:=(a_1+b_0\sqrt{p})^n$ (where we sometimes need the hypothesis $a_1^2-b_0^2p=t^{2k}$ for some integer $k$).  We give a description of the solution for this recurrence in \eqref{equation:explicit-form-an}  and \eqref{equation:explicit-form-bn} and use this solution to obtain multitude of analogues that generalize properties of Chebyshev polynomials of the first and second kinds.  In particular, we discover analogues of a result used in $2$-dimensional potential theory and multipole expansion, an analogue of Tur\'an's inequality, and closed form formula for products of $a_n$'s and $b_n$'s in terms of sums of such polynomials.  This is in addition to finding summation formulae that relate $b_n$ to sums of the $a_n$'s and also a growth formula for the $b_n$'s (cf. Section~\ref{section:properties-results}). 
In Section~\ref{section:new-diff-equations}, under the hypothesis that $a_1^2-b_0^2p=t^{2k}$ for some $k\in\mathbb Z$, we prove a key result that allows one to describe $a_n$ and $b_n$ in terms hypergeometric functions, Jacobi polynomials and Chebyshev polynomials.   In this section one will also find a version of Rodrigue's formulae for the $a_n$ and $b_n$.
The last result in section four one will find second order linear differential equations for which these polynomials $a_n$ and $b_n$ satisfy.   Under the condition that $p$ is the polynomial given in \eqref{eq:polynomial-DJKM}, one can see that one of these differential equations coincides with the one appearing in \cite{coxzhao-2015} (see  \eqref{align:firsttsquaredode}).
Moreover we observe that our differential equations are of Fuchsian type when $p$ is \eqref{eq:polynomial-DJKM}.

Chebyshev polynomials are known to be orthogonal with respect to an appropriate kernel (see \eqref{kernel:firstkind} and \eqref{kernel:secondkind}).   In the setting of $DJKM$-algebras, i.e.,   see \eqref{eq:polynomial-DJKM}, 
there is a three-step recurrence relation \eqref{modifiedbnrecurrence} that resembles those satisfied by orthogonal polynomials.
This suggests that the Laurent polynomials $t^{-n}a_n$ and $t^{-n}b_n$ in the $DJKM$-algebra setting are orthogonal with respect to some kernel.   We prove that this is the case and determine the respective kernels in the last main theorem of our paper, Theorem \ref{secondmaintheorem}. Surprisingly, the orthogonality result obtained also gives
identities of elliptic integrals, which  
is a new development in the theory of elliptic integrals.  The last result we give is a property of the $b_n$'s that is analogous to a property of the Shabat polynomials.

In Section~\ref{section:future-work}, we give suggestions for future work.

\section{Background}\label{section:background}

\subsection{The automorphism group of the derivations of a Riemann surface}\label{subsection:automorphism-Riemann-surface}

In \cite{cox2014simple} where $R_2(p)=\mathbb C[t,t^{-1},u\,|\, u^2=p(t)\}$ is the particular case of where $p(t)$ is a separable polynomial, \color{black} we studied the automorphism group $\Aut(\mathcal{R}_2(p))$ of the derivations of a Riemann surface, 
which are directly related to the units.  This is due to Lemma 3 in \cite{cox2014simple},  where one sees that $\mathcal{R}_m(p)$ and $R_m(p)$ are intricately related by a derivation 
$\mathcal{R}_2(p)= R_2(p)\Delta$, where $\Delta=p'(t)\frac{\partial}{\partial u}+ 2u\frac{\partial}{\partial t}$. Skryabin in \cite{MR966871} and \cite{MR2035385} 
uses this derivation to relate the automorphism groups $\Aut(\mathcal{R}_2(p))$ and $\Aut(R_2(p))$.  Algebras that are central  extensions of Lie algebras of the form $\text{Der}R_2(p)$ and $\mathfrak g\otimes R_2(p)$ we called hyperelliptic Lie algebras due to the fact that $ u^2=p(t)$ is the equation for a hyperelliptic curve. The case where $p(t)$ is the separable polynomial of degree four in \eqref{DJKMalgebra} giving rise to an elliptic curve, we call these 
algebras DJKM algebras.   All these algebras are Krichever-Novikov type algebras. \color{black}

\subsection{The automorphism group of a Riemann surface with a finite number of punctures}\label{subsection:multiplicative-group-in-Riemann-surface}
There is a very large body of literature that has been devoted to the study of the automorphism group a Riemann surface, but we will not review what is known up to this time. In particular it is an interesting question as to which finite groups can appear when fixing the genus of the Riemann surface.  However a reference for such work is \cite{MR1796706}.   On the other hand a result in \cite{cox2014simple} describes the automorphism groups of particular classes of the algebras $R_2(p)$.  
The authors use the fact that an automorphisms sends an invertible element to an invertible element. \color{black}
Since the generators $t$ and $t^{-1}$ of the algebra appear as units, the image of these invertible elements helps to distinguish the distinct automorphisms, but one first needs to describe the group of units, and to that end, we have: 
 
\begin{lemma}[\cite{cox2014simple}, Lemma 11(a)]\label{lemma:multiplicative-group}
The unit group $R_2^*(p)$ of $R_2(p)$ is of the form 
\begin{equation}
\{t^i: i\in \mathbb{Z}\}\cdot \{ f+ g\sqrt{p}: f,g\in \mathbb{C}[t], f^2-g^2 p=ct^k\mbox{ for some }c\in \mathbb{C}^*, k\in \mathbb{Z}_{\geq 0} \}.
\end{equation}
\end{lemma}

For a general polynomial $p\in\mathbb C[t]$, 
an explicit description of all pairs of polynomials $f$ and $g$ that satisfy Pell's equation $ f^2-g^2 p=1$ does not exist. 
  As mentioned earlier, Chebyshev polynomials give all of the solutions when $p(t)=t^2-1$, and we recall from \cite{cox2014simple} the explicit description of $f$ and $g$ for the degree $4$ polynomial in 
  \eqref{eq:polynomial-DJKM} 
 with distinct roots (also see Lemma~\ref{prop:hypergeometric-function} with $r=0$).

\subsection{$DJKM$-algebras} \label{subsection:integrable-systems-Landau-Lifshitz-diff-eq}
Date-Jimbo-Kashiwara-Miwa in \cite{MR701334} and \cite{MR823315} studied certain integrable systems arising from the Landau-Lifshitz differential equation.  In the $DJKM$ setting, we will always assume $p$ has the form $\eqref{eq:polynomial-DJKM}$; it is then clear that   $p(t)=q(t)^2-1$, where $q(t)=\frac{t^2-\beta}{\sqrt{\beta^2-1}}$. 
Let 
\[ 
\begin{aligned}
\lambda_0 &=\lambda_0(\beta,t)= \dfrac{t^2-\beta}{\sqrt{\beta^2-1}} + \sqrt{p}, 	&\lambda_1=\lambda_1(\beta,t) = \dfrac{t^2+1}{\sqrt{2(\beta+1)}} + \sqrt{\dfrac{\beta-1}{2}}\sqrt{p},  \\   
\lambda_2 & =\lambda_2(\beta,t)= \dfrac{t^2-1}{\sqrt{2(\beta-1)}} + \sqrt{\dfrac{\beta+1}{2}}\sqrt{p},  &\lambda_3 =\lambda_3(\beta,t)= \dfrac{\beta t^2-1}{\sqrt{\beta^2-1}} + \sqrt{p}.   \\   
\end{aligned}  
\]  
It is easy to check that $\lambda_i\in R_2^*(p)$ for $0\leq i\leq 3$. Since $\lambda_i$'s are related by the following relations: 
\[ 
	\lambda_0 \overline{\lambda_0}=1, \hspace{2mm} 
	\lambda_1 \overline{\lambda_1} = t^2, \hspace{2mm}
	\lambda_2 \overline{\lambda_2}=t^2, \hspace{2mm}
	\lambda_1 \lambda_2 = t^2 \lambda_0, \hspace{2mm}
	\lambda_1 \overline{\lambda_2} = \lambda_3, 
\]  
$\lambda_1$ and $\lambda_2$ (along with $t$) generate the group of units. This gives us the following: 

\begin{theorem}[\cite{cox2014simple}, Theorem 13(a), 13(b)]\label{theorem:multiplicative-group-generators}
The group $R_2^*(p)$ of units is isomorphic to 
$\mathbb{C}^*\times \mathbb{Z} \times \mathbb{Z}  \times \mathbb{Z} $. 
\end{theorem}
Observe that $\imath\lambda_1(\beta,t)= \lambda_2(-\beta,\imath t)$ and $p_\beta(t)=p_{-\beta}(\imath t)$, so we will focus on powers of $\lambda_2$. Next, we define the families $a_n$ and $b_n$ using the equation
$\lambda_2^n=(a_1+b_0\sqrt{p})^n=a_n+b_{n-1}\sqrt{p}$, where $n\geq 0$. Note that we obtain $b_{-1} = 0$ when $n=0$. 
In Section~\ref{subsection:section2-general-separable-p}, we will derive the recurrence relation for $b_n = b_n(\beta,t)$ satisfying 
$a_n+b_{n-1}\sqrt{p}= (a_1+b_0\sqrt{p})^n$, where the first few terms are 
\begin{equation}\label{eq:initial-conditions}
	a_0 = 1, \hspace{5mm} 
	b_0 = \sqrt{\dfrac{\beta+1}{2}}, \hspace{5mm}  
	a_1 = \dfrac{t^2-1}{\sqrt{2(\beta-1)}}, \hspace{5mm}  
	b_1 = 2 a_1b_0. 
\end{equation}
The families $a_n$ and $b_n$ are closely related to Chebyshev polynomials of the first and the second kind, as we will see below.

\subsection{General separable $p$}\label{subsection:section2-general-separable-p}  For most of the paper, we will assume that $p$ is separable, with the occasional exceptions in the examples noted below. 
So suppose we assume that $a_1$ and $b_0$ are polynomials with complex coefficients and the subsequent polynomials $a_n$ and $b_n$ are defined by the equation $a_n+b_{n-1}\sqrt{p}=(a_1+b_0\sqrt{p})^n$.  Then we have 
\begin{align}
\sum_{n\geq 0}a_nz^n + \sum_{n\geq 0} b_{n-1}\sqrt{p}z^n 
	&= \sum_{n\geq 0} (a_1+b_0\sqrt{p})^n z^n\notag \\  
	&= \dfrac{1-a_1z+b_0\sqrt{p} z}{1-2a_1 z+ (a_1^2-b_0^2 p)z^2} .\label{equation:an-bn}
\end{align}  
The generating series are then given: 
\[ 
\sum_{n\geq 0}a_n z^n = \dfrac{1-a_1 z}{1-2a_1 z + (a_1^2-b_0^2 p)z^2} \hspace{3mm} \mbox{ and }\hspace{3mm} 
\sum_{n\geq 0}b_{n-1} z^n = \dfrac{b_0 z}{1-2a_1 z + (a_1^2-b_0^2 p) z^2} 
\] 
for the pair $a_n$ and $b_n$ of polynomials.    
The second equation gives us  
\[  
	\begin{aligned} 
		\sum_{n\geq 0} b_{n-1}z^n - \sum_{n\geq 0} 2a_1 b_{n-1}z^{n+1} + \sum_{n\geq 0} (a_1^2-b_0^2 p) b_{n-1}z^{n+2} = b_0 z, 
	\end{aligned} 
\]  
which is equivalent to
\[ 
	\begin{aligned}
		\sum_{n\geq 2} b_{n-1}z^n - \sum_{n\geq 1} 2a_1 b_{n-1}z^{n+1} + \sum_{n\geq 0}(a_1^2-b_0^2 p) b_{n-1}z^{n+2}=0,
	\end{aligned}
\] 
or 
\[
\begin{aligned} 
	\sum_{n\geq 0} b_{n+1} z^{n+2} - \sum_{n\geq 0} 2a_1 b_{n}z^{n+2} + \sum_{n\geq 0}(a_1^2-b_0^2 p)  b_{n-1}z^{n+2}=0. 
\end{aligned}
\] 
So we have the recurrence relation
\[ 
	b_{n+1} - 2a_1 b_{n} + (a_1^2-b_0^2 p)  b_{n-1} = 0  \mbox{ for all } n\geq 0, 
\] 
where we assume $b_{-1}=0$.  The $a_n$'s have a similar recurrence relation:
\[ 
	a_{n+2} - 2a_1 a_{n+1} + (a_1^2-b_0^2 p)  a_{n} = 0  \mbox{ for all } n\geq 0, 
\] 
but with the initial condition $a_0=1$.

\subsection{Some second order linear differential equations}

\subsubsection{Chebyshev polynomials} Recall Chebyshev polynomials of the first kind $T_n(t)$ and of the second kind $U_n(t)$, which are defined recursively and satisfy the differential equations
  \begin{align*}
  (1-t^2)y''-t y'+n^2y=0 
  \hspace{3mm} \mbox{ and }  
  \hspace{3mm} 
  (1-t^2 )y''-3t y'+n(n+2)y=0, 
  \end{align*} 
  respectively.  
  They also satisfy the polynomial Pell equation:
  \[ 
T_n(t)^2-(x^2-1)U_{n-1}(t)^2=1.
\] 

\subsubsection{$DKJM$ setting: differential equation for the family $b_n$ of the second kind}\label{subsection:differential-equation-bn}  
We now return to the setting of  
\[ 
p(t) = p_{\beta}(t)=\dfrac{t^4-2\beta t^2+1}{\beta^2-1}, \mbox{ where } \beta\not= \pm 1.
\]

A rather tedious ad-hoc derivation obtained in \cite{cox2014simple} of the second order linear differential equation satisfied by $b_n$'s (in the $DJKM$-setting) gives us
\begin{align}\label{align:firsttsquaredode}
0 &= t(t^2+1)(t^4-2\beta t^2+1)y'' \\ 
&- ((2n-3)t^6+t^4(-4\beta n + 2n-5)+t^2(4\beta-4\beta n+2n+3)+2n+1)y' \notag \\ 
&- 2(2nt^5+nt^3(\beta + (\beta+1)n+5)+nt(-\beta+(\beta+1)n+1)) y.\notag  
\end{align}
This ad-hoc derivation did not yield a low degree linear differential equation satisfied by the $a_n$'s in the setting of $DJKM$-algebras.   However, in Section~\ref{section:new-diff-equations}, we find a more general second order linear differential equation for which the $a_n$'s satisfy.  

\section{Consequences of the recurrence relation}\label{section:properties-results} 
%
%
%
\subsubsection{Separable $p$}

One can easily show by induction that the solution to our recursion relation 
\begin{equation}\label{bnrecursion}
	b_n=2a_1b_{n-1}-(a_1^2-b_0^2p)b_{n-2}  
\end{equation}
is solved by 
\begin{align}\label{equation:explicit-form-bn} 
b_n &= 
	\dfrac{1}{2\sqrt{p}}  
		\left(  
			\left( 
				a_1+ b_0\sqrt{p}
			\right)^{n+1} 
			-
			\left(  
				a_1 - b_0\sqrt{p}
			\right)^{n+1} 
		\right)  
	= b_0 \dfrac{(a_1+b_0\sqrt{p})^{n+1} - (a_1-b_0\sqrt{p})^{n+1}}{(a_1+b_0\sqrt{p})-(a_1-b_0\sqrt{p})},
\end{align}
for $n\geq -1$.

Since 
\[
a_n = (a_1+b_0\sqrt{p})^n - b_{n-1}\sqrt{p}, 
\] 
we also obtain the following formula for the family of polynomials of the first kind: 
\begin{equation}\label{equation:explicit-form-an} 
a_n 	= \dfrac{1}{2}	
		\left(  
			(a_1+b_0\sqrt{p})^n + (a_1-b_0\sqrt{p})^n 
		\right) 
	= 	a_1 \dfrac{(a_1+b_0\sqrt{p})^n + (a_1-b_0\sqrt{p})^n}{(a_1+b_0\sqrt{p})+(a_1-b_0\sqrt{p})}. 
	\color{black}
\end{equation}

The generating function in \eqref{equation:an-bn} gives us the summations in Proposition~\ref{proposition:summations-an-bn}. 

\begin{proposition}\label{proposition:summations-an-bn}
The following holds for the polynomials $a_n$ and $b_n$:
\begin{enumerate}
\item \label{item:ordinary-generating-function}  
\begin{align*}
\sum_{n=0}^{\infty} a_n x^n &=  
	\dfrac{ 1-a_1x 
	}{
		\left(
			1 - x 	\left( 
				a_1+b_0\sqrt{p}	
				\right) 
		\right)
		\left(  
			1 - x	\left(  
				a_1-b_0\sqrt{p}
				\right) 
		\right) 
	}
	= 
	\dfrac{1-a_1x}{1-2a_1 x+(a_1^2-b_0^2 p) x^2}, \\ 
\sum_{n=0}^{\infty} b_n x^n &= 
	\dfrac{ b_0  
	}{  
		\left(
			1 - x 	\left( 
				a_1+b_0\sqrt{p}	
				\right) 
		\right)
		\left(  
			1 - x	\left(  
				a_1-b_0\sqrt{p}
				\right) 
		\right) 
	} 
	= 
	\dfrac{b_0}{1-2a_1 x+(a_1^2-b_0^2 p) x^2}.  \\ 
\end{align*} 
\item \label{item:sum-of-bn}
We have 
\begin{align*} 
\sum_{n=0}^{\infty} a_n &= 
		\dfrac{1-a_1}{		\left(
			1 -  	\left( 
				a_1+b_0\sqrt{p}	
				\right) 
		\right)
		\left(  
			1 - 	\left(  
				a_1-b_0\sqrt{p}
				\right) 
		\right) }, \\  
\sum_{n=0}^{\infty} b_n &=  
		\dfrac{b_0
			}{
		\left( 
			1 - \left( a_1 + b_0\sqrt{p} \right)
		\right)
		\left( 
			1 - \left( a_1 - b_0\sqrt{p} \right)
		\right) 
		}.  \\ 
\end{align*} 
\item \label{item:2-dimensional-potential-theory} 
Generating functions related to $2$-dimensional potential theory and multipole expansion are:  
\begin{align*} 
\sum_{n=1}^\infty a_n\frac{x^n}{n}
		&= -\dfrac{1}{2} \ln 
		\left(  
			(1-(a_1+ b_0\sqrt{p})x)
			(1-(a_1- b_0\sqrt{p})x)
		\right), 
\end{align*}
and 
\begin{align*}
	\sum_{n=1}^{\infty} b_n \dfrac{x^n}{n} 
		&=  \dfrac{1}{2}b_0 	
			\left(  
				\dfrac{a_1}{b_0\sqrt{p}} \ln 
					\left(  
						\dfrac{1- (a_1-b_0\sqrt{p}) x}{1-(a_1+b_0\sqrt{p})x }
					\right)  
					- 
					\ln 
						\left( 
							1-2a_1x+(a_1^2-b_0^2p)x^2
						\right) 
			\right).  \\ 
\end{align*}
\item \label{item:exponential-generating-function}
We have exponential generating functions: 
\begin{align*}
\sum_{n=0}^\infty a_n\frac{x^n}{n!}&=e^{a_1x}\cosh (b_0\sqrt{p}x), \\ 
\sum_{n=0}^{\infty} b_n \dfrac{x^n}{n!} &= 
	b_0e^{a_1 x} 
				\left(\cosh\left( b_0\sqrt{p} x \right)
					+  
					\frac{a_1}{b_0\sqrt{p}}
					\sinh\left( b_0\sqrt{p} x 
						\right)
				\right).  \\ 
\end{align*}
\item \label{item:Turan-inequality}
We have an analogous form of Tur\'an's inequality: 
for $p$ not the square of a polynomial, one has 
\[ 
a_n^2-a_{n-1}a_{n+1} =-pb_0^2\left( a_1^2-b_0^2p
	\right)^{n-1} \not=0, 
\] 
and 
\begin{align*}
b_n^2-b_{n-1}b_{n+1} =b_0^2	\left( a_1+b_0\sqrt{p}
	\right)^n
	\left(a_1-b_0\sqrt{p}		
	\right)^n =b_0^2	\left( a_1^2-b_0^2p	
	\right)^n \not=0.
\end{align*}
\item \label{item:products-of-an-bn}
For $m\geq n$, we give a closed form for products of $a_n$'s and $b_n$'s: 
\begin{align*}
a_ma_n&=\frac{1}{2}\left(a_{m+n} + (a_1^2-b_0^2p)^na_{m-n}\right), 	\\ 
b_ma_n&=\frac{1}{2}\left(b_{m+n} + (a_1^2-b_0^2p)^nb_{m-n}\right),	\\ 
pb_mb_n&=\frac{1}{2}\left(a_{m+n+2}-(a_1^2-b_0^2p)^{n+1}a_{m-n}\right). 	 	\\ 
\end{align*} 
\end{enumerate}  
\end{proposition}

\begin{proof} 
Equation \ref{item:ordinary-generating-function} follows directly from \eqref{equation:an-bn}.

For \ref{item:2-dimensional-potential-theory}, we have 
\[ 
\begin{aligned} 
\sum_{n=1}^{\infty} a_n \dfrac{x^n}{n} 
	&= \dfrac{1}{2} 
		\sum_{n=1}^{\infty}  
			\left( 
				a_1+ b_0\sqrt{p}
			\right)^{n} \dfrac{x^n}{n}
			+
	      \dfrac{1}{2}
		\sum_{n=1}^{\infty}   
			\left(  
				a_1 - b_0\sqrt{p}
			\right)^{n} \dfrac{x^n}{n} \\ 
	&= 	-\dfrac{1}{2} \ln 
		\left(  
			(1-(a_1+ b_0\sqrt{p})x)
			(1-(a_1- b_0\sqrt{p})x)
		\right),   \\ 
		\end{aligned} 
\] 
and 
\[ 
\begin{aligned} 
\sum_{n=1}^{\infty} b_n \dfrac{x^n}{n} 
	&= \dfrac{1}{2\sqrt{p}} 
		\sum_{n=1}^{\infty}  
			\left( 
				a_1+ b_0\sqrt{p}
			\right)^{n+1} \dfrac{x^n}{n}
			- 
	      \dfrac{1}{2\sqrt{p}}
		\sum_{n=1}^{\infty}   
			\left(  
				a_1 - b_0\sqrt{p}
			\right)^{n+1} \dfrac{x^n}{n} \\ 
	&= \dfrac{a_1+ b_0\sqrt{p}}{2\sqrt{p}} 
			\sum_{n=1}^{\infty}  
				\left( 
					a_1+ b_0\sqrt{p}
				\right)^{n} \dfrac{x^n}{n}
				- 
			\dfrac{a_1 - b_0\sqrt{p}}{2\sqrt{p}}
			\sum_{n=1}^{\infty}   
				\left(  
					a_1 - b_0\sqrt{p}
				\right)^{n} 
			\dfrac{x^n}{n} \\ 
	&= -\dfrac{a_1+ b_0\sqrt{p}}{2\sqrt{p}} \ln 
		\left(  
			1-(a_1+ b_0\sqrt{p})x 
		\right) 
		+ 
		\dfrac{a_1- b_0\sqrt{p}}{2\sqrt{p}} \ln 
		\left(  
			1-(a_1- b_0\sqrt{p})x
		\right) \\
	&= 	\dfrac{a_1}{2\sqrt{p}} \ln 
		\left(  
			\dfrac{1-(a_1- b_0\sqrt{p})x}{1-(a_1+ b_0\sqrt{p})x}
		\right)  \\ 
	&\quad -\dfrac{ b_0 }{2 } \ln 
		\left(  
			\left( 1-(a_1+ b_0\sqrt{p})x \right)
			\left(1-(a_1- b_0\sqrt{p})x \right)
		\right).  \\
\end{aligned} 
\] 

Exponential generating functions in \ref{item:exponential-generating-function} follow from 
\[ 
	\begin{aligned}
		\sum_{n=0}^{\infty} a_n \dfrac{x^n}{n!} 
			&= \dfrac{1}{2} \sum_{n=0}^{\infty}  
			\left( a_1+ b_0\sqrt{p}\right)^{n} \dfrac{x^n}{n!}
		+
		\dfrac{1}{2} \sum_{n=0}^{\infty}  
			\left( a_1 - b_0\sqrt{p}\right)^{n} \dfrac{x^n}{n!} \\ 
			&= 
		\dfrac{1}{2} 
		\left( 
			e^{(a_1+b_0\sqrt{p})x}	
			+
			e^{(a_1 - b_0\sqrt{p})x} 
		\right)  \\ 
			&= 
		e^{a_1x}\cosh (b_0\sqrt{p}x), 
	\end{aligned} 
\]  
and 
\[ 
	\begin{aligned}
		\sum_{n=0}^{\infty} b_n \dfrac{x^n}{n!} 
			&= 
		\dfrac{1}{2\sqrt{p}} \sum_{n=0}^{\infty}  
			\left( 
				a_1+ b_0\sqrt{p}
			\right)^{n+1} 
		\dfrac{x^n}{n!}
			-
		\dfrac{1}{2\sqrt{p}} \sum_{n=0}^{\infty}  
			\left(  
				a_1 - b_0\sqrt{p}
			\right)^{n+1} 
		\dfrac{x^n}{n!} \\ 
			&= 
		\dfrac{a_1+ b_0\sqrt{p}}{2\sqrt{p}} \sum_{n=0}^{\infty}  
			\left( 
				a_1+ b_0\sqrt{p}
			\right)^{n} 
		\dfrac{x^n}{n!}
			-
		\dfrac{a_1 - b_0\sqrt{p}}{2\sqrt{p}} \sum_{n=0}^{\infty}  
			\left(  
				a_1 - b_0\sqrt{p}
			\right)^{n} 
		\dfrac{x^n}{n!} \\ 
			&= 
		\dfrac{a_1+ b_0\sqrt{p}}{2\sqrt{p}} 
		e^{(a_1+b_0\sqrt{p})x}	
			-
		\dfrac{a_1 - b_0\sqrt{p}}{2\sqrt{p}}  
		e^{(a_1 - b_0\sqrt{p})x}  \\ 
			&= 
		b_0  
		e^{a_1x}
		\left( 
		\left( 
		e^{b_0\sqrt{p}x}
			+   
		e^{-b_0\sqrt{p}x}
		\right)/2
		\right)
			+
		\dfrac{a_1}{\sqrt{p}} 
		e^{a_1x} 
		\left(  
		\left( 
		e^{ b_0\sqrt{p}x }
			- 
		e^{-b_0\sqrt{p}x}
		\right)/2
		\right) \\ 
		&= e^{a_1 x} 
		\left( 
		b_0\cosh (b_0\sqrt{p}x) + \dfrac{a_1}{\sqrt{p}}\sinh(b_0\sqrt{p}x) 
		\right).   \\ 
	\end{aligned} 
\]

The following proves \ref{item:Turan-inequality}:  
\[   
	\begin{aligned}
		a_n^2  &- a_{n-1}a_{n+1} 
			=\left(\frac{(a_1+b_0\sqrt{p})^{n}+(a_1-b_0\sqrt{p})^{n}}{2 }\right)^2 \\
			&\quad -\frac{\left((a_1+b_0\sqrt{p})^{n-1}+(a_1-b_0\sqrt{p})^{n-1}\right)}{2 }\frac{\left((a_1+b_0\sqrt{p})^{n+1}+(a_1-b_0\sqrt{p})^{n+1}\right)}{2 }
	 \\ 
			&= \frac{(a_1+b_0\sqrt{p})^{2n}+2(a_1+b_0\sqrt{p})^{n}(a_1-b_0\sqrt{p})^{n}+(a_1-b_0\sqrt{p})^{2n}}{4}  \\
			&\quad - \frac{ (a_1+b_0\sqrt{p})^{2n}+\left((a_1+b_0\sqrt{p})^2+(a_1-b_0\sqrt{p})^2\right)(
				a_1^2-b_0^2 p)^{n-1}+(a_1-b_0\sqrt{p})^{2n}}{4}\\ 
			&= \frac{(a_1^2-b_0^2p)^{n} - \left(a_1^2+b_0^2p\right)(a_1^2-b_0^2p)^{n-1}}{2}\\ 
			&=-pb_0^2\left( a_1^2-b_0^2p\right)^{n-1}, 
	\end{aligned}
\] 
and 
\[ 
	\begin{aligned}
		b_n^2 &- b_{n-1}b_{n+1} 
			= 
				\dfrac{1}{4p}   
					\left( 
						\left(  
							\left( 
								a_1+ b_0\sqrt{p}
							\right)^{n+1} 
						-
							\left(  
								a_1 - b_0\sqrt{p}
							\right)^{n+1} 
						\right)^2   \right. \\ 
			&\left. \quad - 
				\left(  
					\left( 
						a_1+ b_0\sqrt{p}
					\right)^{n} 
					-
					\left(  
						a_1 - b_0\sqrt{p}
					\right)^{n} 
				\right)   
			( 
			\left( 
				a_1+ b_0\sqrt{p}
			\right)^{n+2} 
			-
			\left(  
				a_1 - b_0\sqrt{p}
			\right)^{n+2} 
			)
		\right)  \\ 
			&= 
				\dfrac{1}{4p}    
					( 
					(a_1+b_0\sqrt{p})^{n'}
					- (a_1+b_0\sqrt{p})^{n'} 
					+ (a_1 - b_0\sqrt{p})^{n'} 
					- (a_1 - b_0\sqrt{p})^{n'}   \\	
			&\quad 
			+ 
			\left( 
				(a_1+b_0\sqrt{p})^2
				-2(a_1+ b_0\sqrt{p})(a_1- b_0\sqrt{p})
				+ (a_1-b_0\sqrt{p})^2
			\right)  \cdot  \\ 
			&\quad \quad  \cdot  (a_1+b_0\sqrt{p})^n (a_1-b_0\sqrt{p})^{n}
			)\\
			&= b_0^2 (a_1+b_0\sqrt{p})^n (a_1-b_0\sqrt{p})^{n},  \\ 
	\end{aligned}    
\]    
where $n' = 2(n+1)$. 

Finally, we leave the proof of \ref{item:products-of-an-bn} to the reader.
\end{proof}

We give additional properties about $a_n$ and $b_n$: 
\begin{proposition}
\begin{enumerate}
\item \label{lemma:product-an-bn}  
 The following summation formulae hold:
 \begin{align*} 
pb_0b_{n-1} &= (a_n-1)(a_1 -1)  - (a_1^2-b_0^2p -2a_1+1) \sum_{k=0}^{n-1} a_k ,
		\\
b_0a_n  &= b_0+b_{n-1}(a_1-1)- (a_1^2-b_0^2p -2a_1+1)\sum_{k=0}^{n-2}b_k.
 \end{align*}
\item  \label{lemma:rate-of-growth} The growth of $\{b_n\}_n$ is determined as:  
\begin{align*}
\dfrac{b_{2n+1}}{b_n} 
	=  2 \sqrt{p} b_{n} 
		+	 
		2\left( a_1 - b_0\sqrt{p}
		\right)^{n+1} = -2 \sqrt{p} b_{n} 
		+ 2\left( a_1 + b_0\sqrt{p}\right)^{n+1} = 2a_{n+1}.
\end{align*}
\end{enumerate}
\end{proposition}

\begin{proof}
Let $\gamma = a_1+b_0\sqrt{p}$ and $\gamma' = a_1 - b_0\sqrt{p}$. 
The proof of the first equality in \ref{lemma:product-an-bn} is as follows: 
 \begin{align*}
\sum_{k=0}^{n-1} a_k 	&= \dfrac{1}{2}  	\sum_{k=0}^{n-1}  
				\left(  
					\gamma^k + (\gamma')^k  
				\right) 		\\ 
				&= \dfrac{1}{2}\dfrac{\gamma^n -1}{\gamma - 1} + \dfrac{1}{2}\dfrac{(\gamma')^n - 1}{\gamma' - 1} \\
				&= \dfrac{
						(a_1 -1)\dfrac{ \gamma^n+(\gamma')^n}{2}
					 - (a_1-1)
					-p b_0 
						\dfrac{ \gamma^n-(\gamma')^n }{2\sqrt{p}}
					}{(\gamma - 1)(\gamma' - 1)} \\
				&= \dfrac{
						 a_n(a_1-1) - (a_1-1) - pb_0b_{n-1}
						}{(\gamma - 1)(\gamma' - 1)},  \\ 
 \end{align*}   
and the proof of the second equality is: 
\begin{align*}
\sum_{k=0}^{n-1} b_k 
		&= \dfrac{1}{2\sqrt{p}} \sum_{k=0}^{n-1} (\gamma^{k+1} - ( \gamma' )^{k+1}) 
			\\   
		&= \dfrac{\gamma}{2\sqrt{p}} \dfrac{\gamma^n-1}{\gamma - 1} 
		- \dfrac{\gamma'}{2\sqrt{p}} \dfrac{(\gamma')^n-1}{\gamma' - 1}  \\ 
		&= \dfrac{
		a_1 (a_1-1) \dfrac{ \gamma^n-(\gamma')^n }{ 2\sqrt{p} }  
		- b_0  \dfrac{ \gamma^n+(\gamma')^n}{2} 
		}{
		(\gamma - 1)(\gamma' - 1)
		}  -
		\dfrac{ 
		b_0
			\left( 
				pb_0   \dfrac{ \gamma^n-( \gamma' )^n}{2\sqrt{p}}   - 1
			\right) 
		}{
		(\gamma - 1)(\gamma' - 1)
		} 
	\\  
	&= \dfrac{b_{n-1}a_1(a_1-1)-b_0(b_0p b_{n-1}-1)-b_0a_n }{
		(\gamma - 1)(\gamma' - 1) 
		}.  \\ 
\end{align*}

For \ref{lemma:rate-of-growth}, we have 
\[  
\begin{aligned} 
\dfrac{b_{2n+1}}{b_n} 
	&= 	\left( 
				\gamma^{n+1}
				\gamma^{n+1} 
			+ 
				\gamma^{n+1}
			\left( 
				\gamma'
			\right)^{n+1}  
			  - 
				\gamma^{n+1}
			\left( 
				\gamma'   
			\right)^{n+1} 
			-
			\left( 
				\gamma' 
			\right)^{n+1}
 			\left( 
				\gamma'
			\right)^{n+1}
		\right) \big/\\ 
		&\hspace{4mm} \left(
				\gamma^{n+1} 
			- 
			\left( 
				\gamma' 
			\right)^{n+1} \right) \\
	&=	\left( 
				\gamma^{n+1} 
			+
			\left( 
				\gamma' 
			\right)^{n+1}
		\right)  
		\left( 
				\gamma^{n+1} 
			-
			\left( 
				\gamma' 
			\right)^{n+1}
		\right) \big/ \\ 
		&\hspace{4mm} \left( 
				\gamma^{n+1} 
			-
			\left( 
				\gamma' 
			\right)^{n+1}
		\right) \\ 
	&=	
			\gamma^{n+1}
		+ 
		\left( 
			\gamma'   
		\right)^{n+1}  \\ 
	&= 	2 \sqrt{p} b_n + 2 \left( \gamma'  \right)^{n+1}, 
		\\ 
\end{aligned} 
\] 
and to prove the second equality, clear the denominator for $b_n$ in \eqref{equation:explicit-form-bn} to get 
\[ 
2 \sqrt{p} b_n = \gamma^{n+1}-(\gamma')^{n+1}.
\] 
\end{proof}

\section{Second order linear differential equations and polynomials $a_n$ and $b_n$}\label{section:new-diff-equations}
\subsection{The key lemma and its corollaries}
For this section, we still consider $p$ to be a separable polynomial. Recall that the hypergeometric function $_2F_1(a,b;c;z)$ is defined as 
$$
_2F_1(a,b;c;z):=\sum_{n\geq 0}\frac{(a)_n(b)_n}{(c)_nn!}z^n, 
$$
where $(a)_n:=a(a+1)(a+2)\cdots (a+n-1)$ is the rising Pochhammer symbol. 
The following is the key lemma used to prove Theorem~\ref{cor:DEQ}: 
\begin{lemma}\label{prop:hypergeometric-function}
If $a_1^2-b_0^2p=t^{2r}$, then we have a hypergeometric function description of $a_n$ and $b_n$: 
\begin{align*}
a_n&=t^{rn} {_2F_1}\left(-n, n;\frac{1}{2};\frac{1}{2}(1-(a_1/t^r))\right)=t^{rn} T_n(a_1/t^r),  \\
b_n& = b_0t^{rn} (n+1){_2F_1}\left(-n, n+2;\frac{3}{2};\frac{1}{2}(1-(a_1/t^r))\right) =b_0t^{rn}U_n(a_1/t^r), 
\end{align*}
where $T_n(t)$ and $U_n(t)$ are $n$-th Chebyshev polynomials of the first and second kind, respectively. 
\end{lemma}

\begin{proof}
We have
\begin{align*}
a_n 	&= 	\frac{(a_1+b_0\sqrt{p})^{n}+(a_1-b_0\sqrt{p})^{n}}{2 }   \\
	&= \sum_{k=0}^{\floor*{ n/2}}
\binom{n}{2k}(a_1^2-t^{2r})^ka_1^{n-2k} \\
	&=a_1^n \sum_{k=0}^{\floor*{ n/2}}
\binom{n}{2k}(1-(t^{2r}/a_1^2))^k  \\
&=t^{rn} {_2F_1}\left(-n, n;\frac{1}{2};\frac{1}{2}(1-(a_1/t^r))\right) \\
&=t^{rn} T_n(a_1/t^r), 
\end{align*}
and 
\begin{align*}
b_n 	&= 	\frac{(a_1+b_0\sqrt{p})^{n+1}-(a_1-b_0\sqrt{p})^{n+1}}{2\sqrt{p}}   \\
	&=b_0 \sum_{k=0}^{\floor*{ n/2}}
\binom{n+1}{2k+1}(a_1^2-t^{2r})^ka_1^{n-2k} \\
	&=b_0a_1^n \sum_{k=0}^{\floor*{ n/2}}
\binom{n+1}{2k+1}(1-(t^{2r}/a_1^2))^k  \\
	&=b_0 t^{rn}(a_1/t^r)^n \sum_{k=0}^{\floor*{ n/2}}
\binom{n+1}{2k+1}(1-(t^{2r}/a_1^2))^k  \\
&=b_0t^{rn} (n+1){_2F_1}\left(-n, n+2;\frac{3}{2};\frac{1}{2}(1-(a_1/t^r))\right) \\
&=b_0t^{rn} U_n(a_1/t^r).  
\end{align*}
\end{proof}

%
%
%

The next two corollaries provide us with alternative ways to compute the $a_n$ and $b_n$ that do not require their defining recursion relation. 
\color{black}

\begin{corollary}  Provided $a_1^2-b_0^2p=t^{2r}$, the polynomials $a_n$ and $b_n$ are given by the following determinant formulae:
\begin{equation}
a_n(t)=\det\begin{pmatrix} a_1 & t^r & 0& 0 & \cdots & 0 & 0 \\
t^r & 2a_1 & t^r& 0 & \ddots & 0 & 0 \\
0 & t^r & 2a_1& t^r & \ddots & 0 & 0 \\
0 & 0 & t^r& 2a_1 & \ddots & 0 & 0 \\
\vdots  & \ddots  & \ddots& \ddots & \ddots & \ddots & \vdots \\
0  & 0  & 0& 0 & \cdots & 2a_1 & t^r \\
0  & 0  & 0& 0 & \cdots & t^r & 2a_1 \\
\end{pmatrix} 
\text{ and }  \:\:\:\: 
b_n(t)=b_0\det\begin{pmatrix} 2a_1 & t^r & 0& 0 & \cdots & 0 & 0 \\
t^r & 2a_1 & t^r& 0 & \ddots & 0 & 0 \\
0 & t^r & 2a_1& t^r & \ddots & 0 & 0 \\
0 & 0 & t^r& 2a_1 & \ddots & 0 & 0 \\
\vdots  & \ddots  & \ddots& \ddots & \ddots & \ddots & \vdots \\
0  & 0  & 0& 0 & \cdots & 2a_1 & t^r \\
0  & 0  & 0& 0 & \cdots & t^r & 2a_1 \\
\end{pmatrix},\notag
\end{equation}
where the above are $n\times n$ matrices, with $n\geq 1$.
\end{corollary}
\begin{proof}  
This follows from Lemma~\ref{prop:hypergeometric-function} 
and a result of P. L. Nash \cite{MR866596}:
\begin{equation}
T_n(t)=\det\begin{pmatrix}
t & -1 & 0& 0 & \cdots & 0 & 0 \\
-1 & 2t & -1& 0 & \ddots & 0 & 0 \\
0 & -1 & 2t& -1 & \ddots & 0 & 0 \\
0 & 0 & -1& 2t & \ddots & 0 & 0 \\
\vdots  & \ddots  & \ddots& \ddots & \ddots & \ddots & \vdots \\
0  & 0  & 0& 0 & \cdots & 2t & -1 \\
0  & 0  & 0& 0 & \cdots & -1 & 2t \\
\end{pmatrix} 
\text{ and } 
U_n(t)=\det\begin{pmatrix} 
2t & -1 & 0& 0 & \cdots & 0 & 0 \\
-1 & 2t & -1& 0 & \ddots & 0 & 0 \\
0 & -1 & 2t& -1 & \ddots & 0 & 0 \\
0 & 0 & -1& 2t & \ddots & 0 & 0 \\
\vdots  & \ddots  & \ddots& \ddots & \ddots & \ddots & \vdots \\
0  & 0  & 0& 0 & \cdots & 2t & -1 \\
0  & 0  & 0& 0 & \cdots & -1 & 2t \\
\end{pmatrix}.\notag
\end{equation}
\end{proof}

Corollary~\ref{corollary:rodrigues-formula} is an analogue of Rodrigues' formula. 

\begin{corollary}\label{corollary:rodrigues-formula}  If $a_1^2-b_0^2p=t^{2r}$ for some nonnegative integer $r$, then 
\begin{align}
a_n(t)&=\frac{(-1)^nt^{(n-1)r}\sqrt{\pi}\sqrt{t^{2r}-a_1^2}}{2^n(n-\frac{1}{2})!} D^n
\left(
		t^{-2nr+r} 
\left(t^{2r}-a_1^2\right)^{n-\frac{1}{2}}
\right),  \\
b_n(t)&=\frac{(-1)^n(n+1)b_0t^{(n+1)r}\sqrt{\pi} }{2^{n+1}(n+\frac{1}{2})!\sqrt{t^{2r}-a_1^2}} D^n
\left(
		t^{-2nr-r}  
\left(t^{2r}-a_1^2\right)^{n+\frac{1}{2}}
\right),  
\end{align}
where 
$\displaystyle{D=\frac{t^{r+1}}{a'_1t-ra_1}\frac{d}{dt}}$.
\end{corollary}

\begin{proof}
The following are Rodrigues' formulae for the Chebyshev polynomials of the first and second kinds:
\begin{align}
T_n(t)&=\frac{(-1)^n\sqrt{\pi}\sqrt{1-t^2}}{2^n(n-\frac{1}{2})!} \frac{d^n}{dt^n}\left(\left(1-t^2\right)^{n-\frac{1}{2}}\right),  \\
U_n(t)&=\frac{(-1)^n(n+1)\sqrt{\pi} }{2^{n+1}(n+\frac{1}{2})!\sqrt{1-t^2}} \frac{d^n}{dt^n}\left(\left(1-t^2\right)^{n+\frac{1}{2}}\right).  
\end{align}
Using Lemma \eqref{prop:hypergeometric-function}, we get the desired result.
\end{proof}

Next result is also essentially a corollary to Lemma~\ref{prop:hypergeometric-function}, but since it is more significant to us, we view it as a theorem. 
 \subsection{Second order linear differential equations}  
\begin{theorem}\label{cor:DEQ}
If $a_1^2-b_0^2p=t^{2r}$, then 
\begin{align*}
0 &=t^2(t^{2r}-a_1^2) (a_1't-ra_1)  a_n''  \\ 
	&\quad 
	- \left(2 r n t(t^{2r}-a_1^2)(a_1't-ra_1)   + a_1t(a'_1t-ra_1)^2   \right.		\\ 
 	&\left.\qquad + t(t^{2r}-a_1^2)\left(t \left(t a_1''-2r a_1'\right)  +r(r+1)a_1  \right)  \right) a_n'    \\ 
     &\quad 
+ \left(r n (t^{2r}-a_1^2)\left(t \left(t a_1''-2r a_1'\right)+r(r+1) a_1\right)  
+ r n a_1(a'_1t-ra_1)^2     \right. \\ 
&\left.\qquad +n^2\left(t a_1'-ra_1\right)^3 
+ (r n+1) r n(t^{2r}-a_1^2)(a_1't-ra_1) \right) a_n, \\ 
\text{ and } &\\
0 &= t^2 (t^{2r}-a_1^2)(t a_1'-a_1 r) b_n'' \\
&\quad -( 2 n r t(t^{2r}-a_1^2)(ta_1'-a_1 r)^2+3 a_1 t (t a_1'-a_1 r)^2 		\\ 
	&\qquad + t (t^{2r}-a_1^2)(a_1 r (r+1) + t(t a_1''- 2 r a_1'))
	 )  b_n'	\\ 
&\quad +( n(n+2)(ta_1'-a_1 r)^3 
+ 3 a_1 r n (ta_1'-a_1 r)^2 	\\ 
&\qquad 	+r n (t^{2r}-a_1^2)(a_1 r (r+1)+t(ta_1''-2ra_1') )  + nr(rn+1)(t^{2r}-a_1^2)(ta_1'-a_1 r)^2 
)b_n,
\end{align*}

   for all $n\geq 0$, where for the last equality, we assume $b_0$ is a constant. 
\end{theorem}

\begin{proof}
We have 
 \begin{align*}
  (1-z^2)T''_n(z)-z T'_n(z)+n^2T_n(z)&=0, \\
  (1-z^2 )U''_n(z)-3z U'_n(z)+n(n+2)U_n(z)&=0,
  \end{align*} 
  so using $z=a_1(t)/t^r$ and
Lemma~\ref{prop:hypergeometric-function}, 
we obtain the second order differential equation for the $a_n$ given above.  
Similarly, 
using $b_n = t^{r n} b_0 U_n(a_1/t^r)$, we obtain the last second order differential equation for the $b_n$. 
 \end{proof}
\color{black}

%
%
%
%
%
%
%
%
%
%
%
%
%
%
%
%
%
%
%
%
%
%
%
%
%
%
%

\begin{remark}
In the $DJKM$ setting, we obtain a second order linear differential equation that the $a_n$'s satisfy: 
\begin{align}\label{align:2ndsecondode}
0 &= t(t^2+1)(t^4-2\beta t^2+1)y''  \\ 
&- ((1-2 n) t^6+(2 n+3) t^4+t^2 (-4 \beta +(4 \beta -2) n-1)-2 n+1)y' \notag \\ 
&-(\beta +1) n t \left(n t^2+n+t^2-1\right)y.\notag  
\end{align}
We were not able to come up with this differential equation in \cite{cox2014simple}.
\end{remark}

\begin{remark}\label{remark:Fuchsian-type}  In the setting of the $DJKM$-algebra (cf. Equation~\eqref{eq:polynomial-DJKM}), the differential equation in Theorem~\ref{cor:DEQ} for $b_n$'s reduces to Equation~\eqref{align:firsttsquaredode}. Moreover, this differential equation for $b_n$ is of Fuchsian type since the analytic coefficient of 
\begin{align}
0 &= y'' - \dfrac{((2n-3)t^6+t^4(-4\beta n + 2n-5)+t^2(4\beta-4\beta n+2n+3)+2n+1)}{t(t^2+1)(t^4-2\beta t^2+1)}y' \\ 
&- \dfrac{2(2nt^5+nt^3(\beta + (\beta+1)n+5)+nt(-\beta+(\beta+1)n+1))}{t(t^2+1)(t^4-2\beta t^2+1)} y \notag
\end{align}
for $y'$  has (distinct) poles at 
$0, \pm i, \pm \sqrt{\beta \pm \sqrt{\beta^2-1}}$ and the degree of the coefficient polynomial of $y'$ in \eqref{align:firsttsquaredode} is $6$, 
while the degree of the polynomial coefficient of $y$ in \eqref{align:firsttsquaredode} is $5$, which is less than or equal to $12$, as required. 

Similarly, one can see that the differential equation \eqref{align:2ndsecondode} is of Fuchsian type.  
\end{remark}


%

\section{Orthogonality in the $DJKM$ setting.}\label{section:Date-Jimbo-Kashiwara-Miwa}  
The recursion relation for Chebyshev polynomials $T_n=T_n(t)$ of the first kind is the following: 
\[
2tT_n=T_{n+1}+T_{n-1}, 
\]
with initial condition $T_0=1$ and $T_1=x$, and it is known that these Chebyshev polynomials are orthogonal with respect to the kernel 
\begin{equation}\label{kernel:firstkind}
\frac{1}{\sqrt{1-t^2}}.
\end{equation}
Similarly, the recursion relation for Chebyshev polynomials $U_n=U_n(t)$ of the second kind is: 
\[
2tU_n=U_{n+1}+U_{n-1}, 
\]
with initial condition $U_0=1$ and $U_1=2x$, and it is known that these Chebyshev polynomials are orthogonal with respect to the kernel 
\begin{equation}\label{kernel:secondkind}
\sqrt{1-t^2}.
\end{equation}

One may recall Favard's Theorem 
(see \cite{favard} and \cite{MR0481884})\color{black},
which states  given a family of polynomials $p_n$, $n\geq 0$ with $p_0=1$, $p_n$ having degree $n$, and satisfying a three term recurrence relation of the form 
$$
p_{n+1}=(t-c_n)p_n-d_np_{n-1}
$$
where $c_n$ and $d_n$ are complex numbers, then the $p_n$ form a sequence of polynomials that are orthogonal with respect to some linear functional $L$ with $L(1)=1$ and $L(p_mp_n)=\delta_{m,n}$. In the setting of $DJKM$-algebras, we can rewrite the recursion relation of the $b_n$'s, \eqref{bnrecursion}, as
\begin{equation}\label{modifiedbnrecurrence}
2 (a_1/t)t^{-n}b_n = t^{-n-1}b_{n+1}+t^{-n+1}b_{n-1}, 
\end{equation} 
which suggests that the Laurent polynomials $t^{-n}b_n$ are orthogonal with respect to some measure if we view them as functions of $a_1/t$.  

Indeed, we prove that this is true in Theorem~\ref{secondmaintheorem} and Corollary~\ref{cor:kernels-orthogonal-system}. 
\color{black}

Consider the polynomial in \eqref{eq:polynomial-DJKM} studied by Date-Jimbo-Kashiwara-Miwa in \cite{MR701334} and \cite{MR823315}, and note that $p(t)=q(t)^2-1$, where $q(t) = \dfrac{t^2-\beta}{\sqrt{\beta^2-1}}$ and 
$p_{-\beta}(it)=p_{\beta}(t)$. 

This gives us the explicit form: 
\begin{align*}
b_n &=
	\dfrac{1}{2}\sqrt{\dfrac{\beta+1}{2}}
		\sqrt{\dfrac{2(\beta-1)}{t^4-2\beta t^2+1} }\cdot  \\ 
			&\cdot\left( 
				\left( 
					\dfrac{t^2-1}{\sqrt{2(\beta-1)}} 
						+ 
					\sqrt{\dfrac{t^4-2\beta t^2+1}{2(\beta-1)} } 
				\right)^{n+1} - 
				\left( 
					\dfrac{t^2-1}{\sqrt{2(\beta-1)}} 
						- 
					\sqrt{\dfrac{t^4-2\beta t^2+1}{2(\beta-1)} } 
				\right)^{n+1}  
			\right).   
\end{align*}
From the recursion relation \eqref{bnrecursion}, we see that $b_n$ is a polynomial in $t$ of degree $2n$.

\begin{theorem} \label{secondmaintheorem}
For $\beta>1$ real and
\[ 
p(t) =\dfrac{t^4-2\beta t^2+1}{\beta^2-1},  \hspace{6mm} a_1(t) := \dfrac{t^2-1}{\sqrt{2(\beta-1)}}, \hspace{3mm}
\mbox{and} 
\hspace{3mm}
b_0 = \sqrt{\dfrac{\beta+1}{2} }, 
\]
we have the following identities: 
\begin{align}
\int_{t= \frac{\sqrt{\beta
   +1}-\sqrt{\beta -1}}{\sqrt{2}}}^{\frac{\sqrt{\beta
   -1}+\sqrt{\beta +1}}{\sqrt{2}}}t^{-n-m-1}a_na_m (t^2+1) \sqrt{\frac{1-\beta }{t^4-2 \beta  t^2+1}}dt   
	= 
	\begin{cases}
 		  	0			& \mbox{ if } n\not=m, \\ 
		\pi \sqrt{\beta -1}		& \mbox{ if } n =m=0, \\ 
		(\pi /2)\sqrt{\beta -1}	& \mbox{ if } n =m\neq 0, 
	\end{cases} \label{integral1}
\end{align}
and
\begin{align}
\int_{t= \frac{\sqrt{\beta
   +1}-\sqrt{\beta -1}}{\sqrt{2}}}^{\frac{\sqrt{\beta
   -1}+\sqrt{\beta +1}}{\sqrt{2}}}t^{-n-m-3}b_mb_n(t^2+1) \sqrt{\frac{t^4-2 \beta  t^2+1}{1-\beta }}dt   = 
	\begin{cases}
 		  	0	& \mbox{ if } n\not=m, \\ 
		\frac{\pi}{2} (\beta+1)\sqrt{\beta-1}	& \mbox{ if } n =m. 
	\end{cases} \label{integral2}
\end{align}
\end{theorem}

\begin{proof}
We know 
\begin{align*}
\int_{-1}^{1}T_n(z)T_m(z)\frac{1}{\sqrt{1-z^2}} dz& = 
	\begin{cases}
 		  	0	& \mbox{ if } n\not=m, \\ 
		\pi 	& \mbox{ if } n =m=0, \\ 
		\pi /2	& \mbox{ if } n =m\neq 0, 
	\end{cases} \\ \\
\int_{-1}^{1}U_n(z)U_m(z)\sqrt{1-z^2}dz &= 
	\begin{cases}
 		  	0	& \mbox{ if } n\not=m, \\ 
		\pi /2	& \mbox{ if } n =m. 
	\end{cases}
\end{align*}
Then setting $z=a_1(t)/t$, we get 
\begin{align*}
 T_n(a_1/t)T_m(a_1/t)\frac{1}{\sqrt{1-(a_1/t)^2}} d(a_1/t)& = t^{-n-m}a_na_m\frac{|t|}{\sqrt{t^2-a_1^2}} (ta_1'-a_1)t^{-2}dt 
 \end{align*}
by Lemma~\ref{prop:hypergeometric-function}. 
Solving $a_1/t=-1$,  
we have two solutions:  
$\displaystyle{t=-\frac{\sqrt{\beta -1}+\sqrt{\beta+1}}{\sqrt{2}}}$, $\displaystyle{\frac{\sqrt{\beta
   +1}-\sqrt{\beta -1}}{\sqrt{2}}}$ 
and for  $a_1/t=1$, $\displaystyle{t=\frac{\sqrt{\beta -1}-\sqrt{\beta+1}}{\sqrt{2}}}$, 
  $\displaystyle{\frac{\sqrt{\beta
   -1}+\sqrt{\beta +1}}{\sqrt{2}}}$.  For $\beta>1$ real, we have 
   \begin{align*}
  \int_{t= \frac{\sqrt{\beta
   +1}-\sqrt{\beta -1}}{\sqrt{2}}}^{\frac{\sqrt{\beta
   -1}+\sqrt{\beta +1}}{\sqrt{2}}} &t^{-n-m-1}a_n a_m\frac{t^2+1}{\sqrt{\beta -1} \sqrt{\frac{t^4-2 \beta  t^2+1}{1-\beta
   }}}dt  \\ 
   &= \int_{-1}^{1}T_n(z)T_m(z)\frac{1}{\sqrt{1-z^2}} dz
   = 
	\begin{cases}
 		  	0	& \mbox{ if } n\not=m, \\ 
		\pi 	& \mbox{ if } n =m=0, \\ 
		\pi /2	& \mbox{ if } n =m\neq 0. 
	\end{cases} 
\end{align*}
Similarly for the $b_n$'s, 
we have 
\begin{align*}
 &\int_{t= \frac{\sqrt{\beta
   +1}-\sqrt{\beta -1}}{\sqrt{2}}}^{\frac{\sqrt{\beta
   -1}+\sqrt{\beta +1}}{\sqrt{2}}} 
	\dfrac{b_0^{-2}}{2}  t^{-n-m-3}b_nb_m
	\sqrt{\frac{t^4-2 \beta  t^2+1}{1-\beta   }}
	\frac{(t^2+1)}{\sqrt{\beta-1}} 
dt\\
 &=\int_{t= \frac{\sqrt{\beta
   +1}-\sqrt{\beta -1}}{\sqrt{2}}}^{\frac{\sqrt{\beta
   -1}+\sqrt{\beta +1}}{\sqrt{2}}}
	b_0^{-2}  t^{-n-m}b_nb_m\frac{\sqrt{t^2-a_1^2}}{|t|} (ta_1'-a_1)t^{-2}dt  	\\ 
&=\int_{-1}^{1}U_n(z)U_m(z)\sqrt{1-z^2}dz 		\\
&= 
	\begin{cases}
 		  	0	& \mbox{ if } n\not=m, 	\\ 
		\pi /2	& \mbox{ if } n =m.  
	\end{cases}
\end{align*}
\end{proof}

\begin{corollary}\label{cor:kernels-orthogonal-system}  
For $\beta>1$ and $n\geq 0$, the Laurent polynomials $t^{-n}a_n$, respectively, $t^{-n}b_n$,  
	form an orthogonal family on the interval $\displaystyle{ \left[ \frac{\sqrt{\beta
   +1}-\sqrt{\beta -1}}{\sqrt{2}}, \frac{\sqrt{\beta
   -1}+\sqrt{\beta +1}}{\sqrt{2}}\right]}$ with respect to the kernels 
   \begin{align*}
t^{-1}(t^2+1) \sqrt{\frac{1-\beta }{t^4-2 \beta  t^2+1}},\quad \text{ respectively, }
\quad t^{-3}(t^2+1) \sqrt{\frac{t^4-2 \beta  t^2+1}{1-\beta }}.
   \end{align*}
\end{corollary}

\begin{remark}  If $m$ is even and $n$ is odd, or vice versa, then the integrals in \eqref{integral1} and \eqref{integral2} are elliptic integrals.
\end{remark}

\begin{remark}  One can also interpret the orthogonality relations \eqref{integral1} and \eqref{integral2} as an orthogonality relation between pairs of polynomials whereby the kernel depends on the degree of the polynomials.  We thank one of the referees for pointing this out to the authors.   
\end{remark}
\color{black}

Lastly we have

\begin{proposition}\label{lemma}\label{lemma:Shabat-poly}
The $b_n$'s have extrema at the endpoints of the interval $|t|\leq 1$, analogous to the property of Shabat polynomials: for $t=\pm 1$, we have 
\begin{align*}
b_n(1) = b_n(-1) = 
		\dfrac{\left(   -1  \right)^{n/2  }}{2} 
		\sqrt{\dfrac{\beta+1}{2}} 
			\left( 
				1+(-1)^n 
			\right) 
 		= 	\begin{cases}
				(-1)^{n/2}\sqrt{\dfrac{\beta+1}{2}} 	&	\mbox{ if } n \mbox{ even},  \\ 
				\hspace{12mm} 0 	&	\mbox{ if } n \mbox{ odd}. 
			\end{cases}  
\end{align*}
\end{proposition}

\begin{proof} 
We have 
\[ 
\begin{aligned} 
b_n(1) &= \dfrac{1}{2} \sqrt{\dfrac{\beta+1}{2}}
			\left( 
				\left( \sqrt{\dfrac{\beta-1}{-(\beta-1)}}
				\right)^n + 
				\left(   
					- \sqrt{
						\dfrac{\beta-1}{-(\beta-1)} 	
						}
				\right)^n  
			\right) \\ 
	&= \dfrac{1}{2}\sqrt{\dfrac{\beta+1}{2}}
		\left( 
			(-1)^{n/2} 	+ (-1)^n(-1)^{n/2}
		\right) 
\\
	&= 		\dfrac{\left(   -1  \right)^{n/2  }}{2} 
		\sqrt{\dfrac{\beta+1}{2}} 
			\left( 
				1+(-1)^n 
			\right),   \\ 
\end{aligned}   
\]  
so the result holds. 
Similarly, since $b_n(-1)=b_n(1)$, we are done. 
\end{proof}  

\section{Future work}\label{section:future-work}
Since the families $a_n$ and $b_n$ of polynomials are intimately related to Chebyshev polynomials, 
it is natural to study the fullest extent of their analogs. In Section~\ref{section:Date-Jimbo-Kashiwara-Miwa}, we specialized to the $DJKM$-algebra setting to obtain an orthogonality result for $a_n$ and $b_n$.  
Our future work includes generalizing the orthogonality to other polynomials $p$, not necessarily of degree 4, whereby one will most likely obtain new identities in terms of hyperelliptic integrals.

\subsection{Acknowledgements}
Both authors thank the Department of Mathematics at the University of California at Santa Cruz for providing 
conducive work environment during the initial stages of this paper. The first author is partially supported by Simons Collaboration Grant \#319261 and the second author is supported by NSF-AWM Mentoring Grant.   We would also like to thank the referees for providing us with useful references and suggestions for improving the exposition of this paper.

\bibliographystyle{gITR}

\appendix
\def\cprime{$'$} \def\cprime{$'$} \def\cprime{$'$} \def\cprime{$'$}

\end{document}